\RequirePackage{fix-cm}

\documentclass[smallextended,envcountsect]{svjour3}       
\smartqed  

\usepackage[latin1]{inputenc}
\usepackage[english]{babel}
\usepackage{verbatim}
\usepackage{psfrag}

\usepackage{calc}
\usepackage{graphicx,palatino}
\usepackage{amsmath}
\usepackage[misc]{ifsym}
\usepackage{soul}

\usepackage{newlfont}

\usepackage{amssymb}

\usepackage{eufrak}
\usepackage{algorithm}
\usepackage[noend]{algpseudocode}
\usepackage{latexsym}

\usepackage{young}
\usepackage[mathscr]{eucal}

\usepackage{amsfonts}
\usepackage{syntonly}
\usepackage{mathptmx}
\usepackage{subfig}
\usepackage{url}
\usepackage{color}
\usepackage{xcolor}
\usepackage{listings} 
\usepackage{tikz}
\usetikzlibrary{arrows,%
                shapes,positioning}

                \usepackage{tikz}

\usetikzlibrary{arrows,%
                petri,%
                topaths}%
\usepackage{tkz-berge}

\usepackage[mathscr]{eucal}
\makeatletter
\newcommand{\addresseshere}{%
  \enddoc@text\let\enddoc@text\relax
}

\newcommand{\ff}[1]{{\mathbb F}_{#1}}
\newcommand{\ffs}[1]{{\mathbb F}_{#1}^\star}
\newcommand{\ffx}[1]{\ff{#1}[x]}

\DeclareMathOperator{\Tr}{Tr}

\newcommand{\fn}{\mathbb F_{2^n}}

\def\cW{\mathcal W}

\newcommand{\ga}{\alpha}
\newcommand{\gb}{\beta}
\newcommand{\g}{\gamma}
\newcommand{\gd}{\delta}

\newcommand{\gr}{\rho}

\newcommand{\gz}{\zeta}

\newcommand{\go}{\omega}

\newcommand{\NL}{\mathcal{N}\hspace{-0.75 mm}\mathcal L}

\input xy
\xyoption{all}

\begin{document}

\title{Differentially low uniform permutations from known 4-uniform functions
}


\author{
        Marco Calderini 
}


\institute{           M. Calderini\at
               \email{marco.calderini@uib.no}\\
         Department of Informatics, University of Bergen, PB 7803, 5020 Bergen, Norway 
         }

\date{Received: date / Accepted: date}

\maketitle

\begin{abstract}
Functions with low differential uniformity can be used in a block cipher as S-boxes since they have good resistance to differential attacks.
In this paper we consider piecewise constructions for permutations with low differential uniformity. In particular, we give two constructions of differentially 6-uniform functions, modifying the Gold function and the Bracken-Leander function on a subfield.
\end{abstract}

%
%

\section{Introduction}
Let $n$ be a positive integer, we will denote by $\ff{2^n}$ the finite field with $2^n$ elements and its multiplicative group by $\ffs{2^n}$.
Permutation maps defined over $\ff{2^n}$ are used as the S-boxes of some symmetric cryptosystems. So, it is important to construct permutations with good cryptographic properties in order to design a cipher that can resist known attacks. In particular, among these properties we have a low differential {\color{black} and boomerang} uniformity for preventing differential {\color{black} and boomerang} attacks \cite{diff,Wagner}, high nonlinearity for avoiding linear cryptanalysis \cite{lin} and also high algebraic degree to resist higher order differential attacks \cite{hdiff}.

Over a field of even characteristic, the best differential uniformity of a function $F$ is two. Functions achieving this value are called almost perfect nonlinear (APN). 
Many works have been done on the construction of APN functions (see for instance \cite{BBMM11,30,BCCCV18,bud09,bud09-2}).
For odd values of $n$ there are known families of APN permutations; while for $n$ even there exists only one example of APN permutation over $\ff{2^6}$ \cite{dillon} and the existence of others remains an open problem. For ease of implementation, usually, the integer $n$ is required to be even in a cryptosystem.
Therefore, finding permutations with good cryptographic properties over $\ff{2^{n}}$ with $n$ even is an interesting research topic for providing more choices for the S-boxes.

The construction of low differentially uniform permutations with the highest nonlinearity over $\ff{2^{n}}$ (with $n$ even) is a difficult task. In Table \ref{tab:4diff} we give five families of primarily constructed differentially 4-uniform permutations with the best known nonlinearity. For all these primarily constructed permutations, but the Kasami function, the boomerang uniformity as been determined \cite{BOCA,calvil,mesn}.
Note that amongst these functions, the Gold, Bracken-Leander and the Bracken-Tan-Tan functions have algebraic degrees 2 or 3.

\begin{table}\caption{Primarily-constructed differentially 4-uniform permutations over $\ff{2^{n}}$ ($n$ even) with the best known nonlinearity}\label{tab:4diff}
\renewcommand{\arraystretch}{1.3}
\centering
\begin{tabular}{|c|c|c|c|c|}
\hline
{\bf Name}& $\mathbf{F(x)}$ & {\bf deg} & {\bf Conditions} & {\bf In}\\
\hline
Gold & $x^{2^i+1}$ & $2$ & $n=2k$, $k$ odd $\gcd(i,n)=2$&\cite{gold}\\
\hline
Kasami & $x^{2^{2i}-2^i+1}$ &i+1& $n=2k$, $k$ odd $\gcd(i,n)=2$&\cite{kasami}\\
\hline
Inverse& $x^{2^n-2}$ & $n-1$ & $n=2k$, $k\ge 1$ &\cite{goldwalsh}\\
\hline
Bracken-Leander &  $x^{2^{2k}+2^k+1}$ & $3$ & $n=4k$, $k$ odd & \cite{BL10} \\
\hline
 &   &  & $n = 3m$, $m$ even, $m/2$ odd, & \\
 Bracken-Tan-Tan&  $\gz x^{2^i+1} +\gz^{2^m} x^{2^{-m}+2^{m+i}}$ & $2$  &   $\gcd(n, i) = 2$, $3|m+i$ &\cite{BTT} \\
  &   &   & and $\gz$ is a primitive element of $\ff{2^n}$ & \\
\hline
\end{tabular}
\end{table}

In the last years, many works on permutations with low differential uniformity have been done.
Several constructions of differentially 4-uniform permutations have been found by modifying the inverse function on some subsets of $\ff{2^n}$ \cite{PHT17,17,19,QTTL,TCT,22,amc19,ZLS}. 
For example, some new constructions of differentially 4-uniform functions were obtained by shifting the inverse function on some subsets of $\ff{2^n}$ \cite{QTTL,TCT}. While in \cite{PHT17,amc19,ZLS} the authors change the inverse function on some subfields of $\ff{2^n}$. Moreover, all the functions constructed in those references were proved to have the optimal algebraic degree ($n-1$) and high nonlinearity.

{For the case of 6-uniform permutations, in \cite{BPpow} the authors studied the differential spectra of some power functions. In \cite{CKS}, the authors study some particular sparse permutation polynomials showing, in particular, that we can obtain 6-uniform permutation polynomials of type $x^{-1}+\g \Tr(x^r)$.
In \cite{6diffe}, the authors give two examples of 6-uniform permutations modifying the Gold function.}

In this paper, we investigate the piecewise construction as in \cite{PHT17,amc19,ZLS} by modifying the image of other well-known 4-uniform permutations on some subfields of $\ff{2^n}$. In particular, we consider the case of the Gold and of the Bracken-Leander functions. We show that in these cases it is possible to obtain permutations with differential uniformity at most 6. Moreover, if we modify these functions using the inverse function (or a function equivalent to it) on a subfield of $\ff{2^n}$, then we can obtain permutations with algebraic degree $n-1$ (which is the highest possible) and high nonlinearity. These results extend those given in \cite{6diffe}. 

The paper is organized as follows. In Section \ref{sec:1}, we give some notations and preliminaries. In Section \ref{sec:2}, we give just a result for obtaining 4-uniform permutation using the piecewise construction with APN permutations.  
In Section \ref{sec:3}, we report our study on the piecewise construction coming from the Gold and Bracken-Leander functions. We show that for these functions the construction can produce a permutation with differential uniformity at most 6. Moreover, when we use the inverse function on the subfield, the obtained function has also algebraic degree $n-1$ and high nonlinearity.
{\color{black} In the last section, we give some results on the boomerang uniformity for some type of piecewise permutations. In particular, we give an upper bound on the boomerang uniformity of piecewise permutations obtained using a 4-uniform Gold function. }

\section{Preliminaries}\label{sec:1}

Any function $F$ from $\ff{2^n}$ to itself can be represented as a univariate polynomial of degree at most $2^n-1$, that is
$$F(x)=\sum_{i=0}^{2^n-1}a_ix^i.$$ 
The {\em 2-weight} of an integer $0\le i\le 2^n-1$, denoted by $w_2(i)$, is the (Hamming) weight of its binary representation. It is well known that the algebraic degree of a function $F$ is given by
$$\deg(F)=\max\{w_2(i)\,\mid\, a_i\ne 0\}.$$
Functions of algebraic degree $1$ are called {\em affine}. Linear functions
are affine functions with constant term equal to zero and they can be represented as $L(x)=\sum_{i=0}^{n-1} a_i x^{2^i}$. For any permutation $F:\ff{2^n}\to\ff{2^n}$ it is well known that the algebraic degree can be at most $n-1$.

For any $m\ge 1$ such that $m|n$ we can define the (linear) \emph{trace function} from $\ff{2^n}$ to $\ff{2^m}$ by 
$$
\Tr_m^n(x)=\sum_{i=0}^{n/m-1}x^{2^{im}}.
$$
When $m=1$ we will denote $\Tr_1^n(x)$ by $\Tr(x)$.

For any function $F:\ff{2^n}\to\ff{2^n}$ we denote the {\em Walsh transform} in $a,b\in \ff{2^n}$ by 
$$
\cW_F(a,b)=\sum_{x\in \ff{2^n}}(-1)^{\Tr(ax+bF(x))}.
$$ 

With {\em Walsh spectrum} we refer to the set of all possible values of the Walsh transform. The Walsh spectrum of a vectorial Boolean function $F$ is strictly related to the notion of nonlinearity of $F$, denoted by $\NL(F)$, indeed we have
$$
\NL (F)=2^{n-1}-\frac{ 1}{2} \max_{a\in\ff{2^n},b\in\ffs{2^n}}|\cW_F(a,b)|.
$$

When $n$ is odd, it has been proved that $\NL(F) \le 2^{n-1}-2^{\frac{n-1}{2}}$;
for $n$ even, the best known nonlinearity is 
$2^{n-1}-2^{\frac{n}{2}}$, and it is conjectured that $\NL(F) \le 2^{n-1}-2^{\frac{n}{2}}$.
All the functions in Table \ref{tab:4diff} reach this value.

The concept of differential uniformity of a function $F$ is related to the number of solutions of the equation $F(x+a)+F(x)=b$ for $a\in\ffs{2^n}$ and $b\in\ff{2^n}$.
\begin{definition}
For a function $F$ from $\ff{2^n}$ to itself, and any $a\in\ffs{2^n}$ and $b\in\ff{2^n}$, we denote by $\gd_F (a, b)$ the number of solutions of the equation $F (x + a) + F (x) = b$. The maximum $$\gd_F=\max_{a\in\ffs{2^n},b\in\ff{2^n}}\gd_F(a,b)$$ is called the differential uniformity of $F$, and F is said to be differentially $\gd_F$-uniform.
\end{definition}

A function F is called almost perfect nonlinear (APN) if $\gd_F=2$. Note that the possible minimum value of $\gd_F$ is 2, since if $x$ is a solution of  $F (x + a) + F (x) = b$, then  $x + a$ is also a solution of the equation.

{\color{black} In \cite{cid}, Cid et al. introduced the concept of Boomerang Connectivity Table for a permutation $F$ over $\fn$. Next, in \cite{BOCA} the authors introduced the notion of boomerang uniformity.
\begin{definition}
Let $F$ be a permutation over $\ff{2^n}$, and $a,b$ in $\ff{2^n}$.
The Boomerang Connectivity Table (BCT) of $F$ is given by a $2^n \times 2^n$ table $T_F$, in which the entry for the position $(a, b)$ is given by
$$T_F (a, b) = |\{x\in\fn\,:\, F^{-1}(F (x) + a) + F^{-1}(F (x + b) + a) = b\}|.$$
Moreover, the value
$$\gb_F = \max_{a,b\in\ffs{2^n}}|\{x\in\fn\,:\, F^{-1}(F (x) + a) + F^{-1}(F (x + b) + a) = b\}|$$
is called the boomerang uniformity of $F$, or we call $F$ a {\em boomerang $\gb_F$-uniform} function.
\end{definition}
In \cite{cid}, the authors show that $\gd_F\le\gb_F$ for any function $F$. Moreover, $\gd_F=2$ if and only if $\gb_F=2$. So, APN permutations offer an optimal resistance to both differential and boomerang attacks.
}

There are several equivalence relations of functions for which the differential uniformity and the Walsh spectrum (and in particular the nonlinearity) are preserved.
Two functions $F$ and $F'$ from $\ff{2^n}$ to itself are called:
\begin{itemize}
\item {\it affine equivalent} if $F' = A_1 \circ F \circ A_2$ where the mappings $A_1,A_2:\ff{2^n}\to\ff{2^n}$ are affine permutations;
\item {\it extended affine equivalent} (EA-equivalent) if $F' =F'' + A$, where the mappings $A:\ff{2^n}\to\ff{2^n}$ is affine and
$F''$ is affine equivalent to $F$;
\item {\it Carlet-Charpin-Zinoviev equivalent} (CCZ-equivalent) if for some affine permutation $\mathcal L$ of $\ff{2^n}\times \ff{2^n}$ the image of the graph of $F$ is the graph of $F'$, that is, $\mathcal L(G_F) = G_{F'}$, where $G_F = \{(x,F(x)) \,:\, x \in\ff{2^n}\}$ and $G_{F'} = \{(x,F'(x)) \,:\, x \in\ff{2^n}\}$.
\end{itemize}

Obviously, affine equivalence is included in the EA-equivalence, and it is also well known that EA-equivalence is a particular case of CCZ-equivalence and every permutation is CCZ-equivalent to its inverse \cite{ccz}.

The algebraic degree is invariant for the affine equivalence and also for the EA-equivalence for nonlinear functions, but not for the CCZ-equivalence. The boomerang uniformity is preserved by  affine equivalence and inverse transformation, but not from  EA- and CCZ-equivalence in general \cite{BOCA}.\\

Some secondary construction methods have been introduced to find new low differentially uniform functions from the known ones. 
For example, some constructions of differentially 4-uniform permutations over $\ff{2^{2m}}$, of degree $m+1$, by using Gold APN functions over $\ff{2^{2m+1}}$, were obtained by Li and Wang in \cite{LW}, inspired by the idea introduced by Carlet in \cite{C11}. 
In \cite{CTTL}, Carlet et al. give a construction of a differentially 4-uniform function over $\ff{2^{2m}}$ from the inverse permutation over $\ff{2^{2m-1}}$.
In the following we will denote the inverse function by $x^{-1}$ (with $0^{-1}:=0$).

The inverse function is APN over $\ff{2^n}$ when $n$ is odd, and is differentially 4-uniform when $n$ is even \cite{goldwalsh}. In the recent years, it has been found that a large number of differentially 4-uniform permutations can be obtained by modifying the inverse function on some subset of $\ff{2^n}$. Among these constructions we have

\begin{itemize}
\item The function $$
F_1(x)=\begin{cases}
x^{-1}+1&\mbox{ if $x\in S$}\\
x^{-1}&\mbox{ if $x\notin S$}
\end{cases}
$$
where $S$ is some specific subset of $\ff{2^n}$ ($n$ even). Some classes of differentially 4-uniform permutations of this form are given in \cite{17,19,QTTL,22,ZZ}
\item The function $$
F_{t_1,t_2}(x)=\begin{cases}
t_1x^{-1}+t_2&\mbox{ if $x\in \ff{2^s}$}\\
x^{-1}&\mbox{ if $x\notin \ff{2^s}$}
\end{cases}
$$
where $n=sm$ with $s$ even and $n/s$ odd, $t_1,t_2\in\ff{2^s}$ and $t_1\ne 0$ \cite{ZLS}.
\item The function $$
F_{\ga,\gb}(x)=\begin{cases}
\gb(x+1)^{-1}+\ga&\mbox{ if $x\in \ff{2^s}$}\\
x^{-1}&\mbox{ if $x\notin \ff{2^s}$}
\end{cases}
$$
where $n=sm$ with $s$ even and $n/s$ odd, $\ga,\gb\in\ff{2^s}$ and $\gb\ne 0$ \cite{PHT17}.
\item The function $$
F_\g(x)=\begin{cases}
(\g x)^{-1}&\mbox{ if $x\in U$}\\
x^{-1}&\mbox{ if $x\notin U$}
\end{cases}
$$
where $\Tr(\g)=\Tr(\g^{-1})=1$ and  $U$ is some specific subset of $\ff{2^n}$ ($n$ even such that $n/2$ is odd) \cite{ptw}.

\item The function $$
F_{\g}(x)=\begin{cases}
(\g x)^{-1}&\mbox{ if $x\in \ff{2^s}$}\\
x^{-1}&\mbox{ if $x\notin \ff{2^s}$}
\end{cases}
$$
where $n=sm$ with $s$ even and $n/s$ odd for any $\g\in\ffs{2^s}$, and with also $s$ odd if $\Tr(\g)=0$  \cite{amc19}.
\end{itemize}

In the next sections, we will study the piecewise construction for the case of Gold and Bracken-Leander 4-uniform permutations and we will show that it is possible to obtain permutations with low differential uniformity.
 
\section{Differentially 4-uniform piecewise functions from APN functions}\label{sec:2}

If we use APN permutations in the piecewise construction, then we can obtain a differentially 4-uniform function. This applies in particular to the case of $n$ odd for which some families of  APN permutations are known. 

\begin{proposition}\label{prop:apn4}
Let $n=sm$, with $m$ odd. Let $f$ be an APN permutation over $\ff{2^s}$ and $g\in \ff{2^s}[x]$ an APN permutation over $\ff{2^n}$. Then, the function
$$F(x)=f(x)+(f(x)+g(x))(x^{2^s}+x)^{2^n-1}=\begin{cases}
f(x)&\mbox{ if $x\in\ff{2^s}$}\\
g(x)&\mbox{ if $x\notin\ff{2^s}$}
\end{cases}$$
is a differentially 4-uniform permutation.
\end{proposition}
\begin{proof}
We need to check that for any $a\in\ffs{2^n}$ and $b\in \ff{2^n}$ the equation
\begin{equation}\label{eq:der}
F(x)+F(x+a)=b
\end{equation}
admits at most 4 solutions.
First consider $a\in\ffs{2^s}$. Then we can have the equations
$$
f(x)+f(x+a)=b
$$
and
$$
g(x)+g(x+a)=b
$$
if $x\in\ff{2^s}$ or $x\notin\ff{2^s}$. In both cases we can have at most 2 solutions which implies that we can have at most 4 solutions for \eqref{eq:der}.

When $a$ is not in $\ff{2^s}$, then we can have two cases:
\begin{itemize}
\item a solution $x$ is in $\ff{2^s}$ and $x+a$ not;
\item both the solutions $x$ and $x+a$ are not in $\ff{2^s}$.
\end{itemize}
Let us count the solutions of first type, that is, we want to count the number of pairs in $\ff{2^s}\times(a+\ff{2^s})$ which are solutions of our equation. To do that, we can count the number of solutions which belong to $a+\ff{2^s}$. Then, consider $x\in a+\ff{2^s}$, from \eqref{eq:der} we obtain 
\begin{equation}\label{eq:dermista}
g(x)+f(x+a)=b.
\end{equation}
Raising \eqref{eq:dermista} to the power $2^s$ and adding the result to \eqref{eq:dermista} we have
$$
g(x)+g(x^{2^s})=b^{2^s}+b,
$$
recall that $g(x)^{2^s}=g(x^{2^s})$ since $g\in\ff{2^s}[x]$. Moreover, since $g$ is a permutation over $\ff{2^n}$, for $b\in\ff{2^s}$ we have no solution of this type. If $b\notin\ff{2^s}$, then we have $x^{2^s}=x+c$ whit $c=a^{2^s}+a\ne 0$. Thus, the equation $g(x)+g(x^{2^s})=b^{2^s}+b$ admits at most two solutions $x$ and $x+c$ (recall that $g$ is APN over $\ff{2^n}$). Since $x\in  a+\ff{2^s}$, we have that $x+c\in a^{2^s}+\ff{2^s}$ and $a^{2^s}+\ff{2^s}\cap a+\ff{2^s}=\emptyset$. Indeed, if $a^{2^s}+\ff{2^s}\cap a+\ff{2^s}\ne\emptyset$, then $a^{2^s}+a\in\ff{2^s}$ and thus $a^{2^{2s}}=a$, implying $\ff{2^{2s}}\subseteq \ff{2^n}$, which is not possible.
This implies that we can have at most one solution of \eqref{eq:dermista} in $ a+\ff{2^s}$. This leads to at most two solutions
of  \eqref{eq:der}.

For the second case, we have that \eqref{eq:der} is given by
$$
g(x)+g(x+a)=b,
$$
which admits at most 2 solutions. This implies that for $a\notin\ff{2^s}$ Equation \eqref{eq:der} admits at most 4 solutions.
\qed\end{proof}
\begin{remark}
Note that if in Proposition \ref{prop:apn4} one APN function is a permutation and the other one no, then the piecewise function is still differentially 4-uniform but it may not be a permutation. In particular, the restriction on $m$ odd is necessary for having $g$ to be an APN permutation (see for instance \cite{hou}).
\end{remark}

\section{Differentially 6-uniform permutations from the Gold and Bracken-Leander functions}\label{sec:3}

In this section we will study the piecewise construction for the case of Gold and the Bracken-Leander function. Before to consider these functions, we will give a general result on the differential uniformity of some piecewise functions.

\begin{theorem}\label{th:main}
Let $n=sm$ for some positive integers $s$ and $m$. Let $f$ and $g$ be two polynomials with coefficients in $\ff{2^s}$, that is $f,g\in\ff{2^s}[x]$, and $g$ permuting $\ff{2^n}$. Suppose that:
\begin{itemize}
\item[(H1)] $f$ is a $\gd_f$-uniform function over $\ff{2^s}$;
\item[(H2)] $g$ is a $\gd_g$-uniform function over $\ff{2^n}$;
\item[(H3)] for any $a\in\ffs{2^s}$ and $b\in\ff{2^s}$ the equation $g(x)+g(x+a)=b$ has no solution in $\ff{2^n}\setminus\ff{2^s}$.
\end{itemize}
Then, the function $$F(x)=f(x)+(f(x)+g(x))(x^{2^s}+x)^{2^n-1}=\begin{cases}
f(x)&\mbox{ if $x\in\ff{2^s}$}\\
g(x)&\mbox{ if $x\notin\ff{2^s}$}
\end{cases}$$
is such that
$$
\gd_F(a,b)\le\begin{cases}
\max\{\gd_f,\gd_g\}&\mbox{ if $a\in\ff{2^s}$}\\
\gd_g+2&\mbox{ if $a\notin\ff{2^s}$}.
\end{cases}$$
\end{theorem}
\begin{proof}
We need to check the number of solutions of the equation
\begin{equation}\label{eq:derF}
F(x)+F(x+a)=b.
\end{equation}
Suppose that $a\in\ffs{2^s}$. Then, we can have that both the solutions $x$ and $x+a$ are in $\ff{2^s}$ or none is in $\ff{2^s}$. In the first case, \eqref{eq:derF} becomes
$$
f(x)+f(x+a)=b,
$$
which has at most $\gd_f$ solutions if $b\in\ff{2^s}$ and none when $b\notin\ff{2^s}$. 

In the second case, we have the equation
$$
g(x)+g(x+a)=b.
$$
From (H3) we have no solution in $\ff{2^n}\setminus\ff{2^s}$ if $b\in\ff{2^s}$. 
If $b\notin\ff{2^s}$ we can have at most $\gd_g$ solutions. Then, for $a\in\ffs{2^s}$ we have at most $\gd=\max\{\gd_f,\gd_g\}$ solutions for Equation~\eqref{eq:derF} for any $b$.

Consider, now, $a \notin \ff{2^s}$. We can have two cases:
\begin{itemize}
\item[(i)] a solution $x$ is in $\ff{2^s}$ and $x+a$ not;
\item[(ii)] both the solutions $x$ and $x+a$ are not in $\ff{2^s}$.
\end{itemize}
We want to count the number of pairs $(x,x+a)$ in $\ff{2^s}\times(a+\ff{2^s})$ which are solutions of Equation \eqref{eq:derF}. Without loss of generality, we can suppose $x\in a+\ff{2^s}$. Then  \eqref{eq:derF} becomes
\begin{equation}\label{eq:dermistaF}
g(x)+f(x+a)=b.
\end{equation}
From this, raising \eqref{eq:dermistaF} by $2^s$ and adding the result to \eqref{eq:dermistaF}, we obtain the equation $g(x)^{2^s}+g(x)=b^{2^s}+b$. 
Denoting by $y=g(x)$, we obtain 
$$
y^{2^s}+y=b^{2^s}+b.
$$
The solutions of this last equation are the elements of the coset $b+\ff{2^s}$. Since we supposed that $x\in a+\ff{2^s}$ and $g$ permutes $\ff{2^n}$ we need to check how many elements we have in $g(a+\ff{2^s})\cap (b+\ff{2^s})$, where $g(a+\ff{2^s}):=\{g(x)\,:\,x\in a+\ff{2^s}\}$.
Suppose that $|g(a+\ff{2^s})\cap (b+\ff{2^s})|\ge 2$. Then, there exist $z_1,z_2,w_1,w_2\in\ff{2^s}$ such that $b+z_1=g(a+w_1)$, $b+z_2=g(a+w_2)$ and $z_1\ne z_2$, $w_1\ne w_2$. Thus, 
$$g(a+w_1)+g(a+w_2)=z_1+z_2.$$ Denoting by $x'= a+w_1$ and $a'=w_1+w_2$, we obtain that 
$$g(x')+g(x'+a')=z_1+z_2.$$ 
Thus we would have a solution $x'\notin\ff{2^s}$ for the equation $g(x)+g(x'+a')=z_1+z_2$, which is not possible by (H3). Therefore, $|g(a+\ff{2^s})\cap (b+\ff{2^s})|\le 1$, implying that in $\ff{2^s}\times(a+\ff{2^s})$ we have at most one pair $(x,x+a)$ which can be solutions of \eqref{eq:derF}.

For the second case we have the equation
$$
g(x)+g(x+a)=b,
$$
which admits at most $\gd_g$ solutions for any $b$. Then, in total for the case $a\notin\ff{2^s}$ Equation \eqref{eq:derF} admits at most $\gd_g+2$ solutions.
\qed\end{proof}

It is easy to note that if also $f$ permutes $\ff{2^s}$, then $F$ is a permutation over $\ff{2^n}$.
\begin{remark}
Note that Proposition \ref{prop:apn4} and Theorem \ref{th:main} use a similar analysis. However, the proof of Proposition \ref{prop:apn4} relies strictly on the APN property of $g$ for the case $x\in\ff{2^s}$ and $x+a\notin\ff{2^s}$ and not on Condition (H3). So we cannot include the result of Proposition \ref{prop:apn4} into Theorem \ref{th:main}.
\end{remark}

\begin{remark}
When the function $g$ in Theorem~\ref{th:main} is a power function, i.e. $g(x)=x^d$, Condition (H3) can be checked only for $a=1$. Indeed, for any nonzero $a$ we have that studying
$$
x^d+(x+a)^d=b
$$
is equivalent to study
$$
\left(\frac x a\right)^d+\left(\frac x a+1\right)^d=\frac{b}{a^d}.
$$

Thus, for the case of Gold and Bracken-Leander functions we need to check if the equation $x^d+(x+1)^d=b$ has no solution in $\ff{2^n}\setminus\ff{2^s}$ whenever $b$ is in $\ff{2^s}$.
\end{remark}
For the Gold function we have the following.

\begin{lemma}\label{lm:1}
Let $n=sm$ with $s$ even and $m$ odd. Let $k$ be such that $\gcd(k,n)=2$. For any $b\in \ff{2^s}$ the equation 
$$
x^{2^k}+x=b
$$
does not admit any solution $x$ in $\ff{2^n}\setminus\ff{2^s}$.
\end{lemma}
\begin{proof}
Suppose that there exists $b\in\ff{2^s}$ for which the equation admits a solution $x\notin\ff{2^s}$. Then, we have that $x^{2^k}+x=x^{2^{k+s}}+x^{2^s}$, 
which also implies $(x^{2^s}+x)^{2^k}=x^{2^s}+x$. Then, $x^{2^s}+x\in\ff{2^n}\cap\ff{2^k}=\ff{4}\subset\ff{2^s}$. Since $m$ is odd and $x\notin\ff{2^s}$ this is not possible. Indeed, $x^{2^s}+x\in\ff{2^s}$ implies $(x^{2^s}+x)^{2^s}=x^{2^s}+x$ and thus $x^{2^{2s}}=x$, and then $x\in\ff{2^n}\cap\ff{2^{2s}}=\ff{2^s}$.
\qed\end{proof}

Thus, we have immediately the following result.
\begin{theorem}\label{th:6diff1}
Let $n=sm$ with $s$ even such that $s/2$ and $m$ are odd. Let $k$ be such that $\gcd(k,n)=2$ and $f$ be at most differentially 6-uniform permutation over $\ff{2^s}$. Then
$$F(x)=f(x)+(f(x)+x^{2^k+1})(x^{2^s}+x)^{2^n-1}=\begin{cases}
f(x)&\mbox{ if $x\in\ff{2^s}$}\\
x^{2^k+1}&\mbox{ if $x\notin\ff{2^s}$}
\end{cases}$$
is a differentially 6-uniform permutation over $\ff{2^n}$.
\end{theorem}
\begin{proof}
Lemma~\ref{lm:1} implies that we have no solution in $\ff{2^n}\setminus\ff{2^s}$ of $x^{2^k}+(x+1)^{2^k}=b$ when $b\in\ff{2^s}$. Then, we have our claim from Theorem~\ref{th:main}.
\qed\end{proof}

For the Gold function it was considered in \cite{6diffe} the case of changing its image on the subfield $\ff{4}$ composing the function with the cycle $(1,\go,\go^2)$, where $\go^2=\go+1$; or  modifying it using $f(x)=x^{2^k+1}+1$ over a subfield $\ff{2^s}$, with $s/2$ odd. Theorem~\ref{th:6diff1} generalizes the results given in \cite{6diffe}.

Also for the case of Bracken-Leander function it is possible to prove a similar result as in Lemma~\ref{lm:1}.
\begin{lemma}\label{lm:4k}
Let $n=4k=sm$ with $k$ and $m$ odd. For any $b\in \ff{2^s}$ the equation 
\begin{equation}\label{eqn:1}
x^{2^{2k}+2^k} +x^{2^{2k}+1} +x^{2^k+1} +x^{2^{2k}} +x^{2^k} +x =b
\end{equation}
does not admit any solution $x$ in $\ff{2^n}\setminus\ff{2^s}$.
\end{lemma}
\begin{proof}
The proof is obtained following similar steps to those used in \cite[Theorem~1]{BL10} for proving the differential 4-uniformity of the map $x^{2^{2k}+2^k+1}$.

Denoting by $\Tr^{4k}_k$ the trace map from $\ff{2^{4k}}$ to $\ff{2^k}$, since  $$\Tr^{4k}_k (x^{2^{2k}+2^k}+x^{2^{2k}+1}+x^{2^k+1}+x^{2^{2k}}+x^{2^k})=0,$$  from \eqref{eqn:1} we obtain $\Tr^{4k}_k(x)= x^{2^{3k}}+x^{2^{2k}}+x^{2^{k}}+x=\Tr^{4k}_k(b)=c$ and thus $x^{2^{2k}}+x^{2^k}=x^{2^{3k}}+x+c$. From this, Equation \eqref{eqn:1} can be rewritten as
\begin{equation}\label{eqn:2}
x^{2^{2k}+2^k} +x(x^{2^{3k}}+x+c) +x^{2^{3k}} =c+b
\end{equation}

Now, using always the fact that  $x^{2^{3k}}+x^{2^{2k}}+x^{2^{k}}+x=c$, we have that raising \eqref{eqn:2} to the power ${2^{2k}}$ and adding the result to Equation \eqref{eqn:2} we obtain
\begin{equation}\label{eqn:3}
 (x+x^{2^{2k}})^2+(c+1)(x+x^{2^{2k}})=b^{2^{2k}}+b+c=b'.
\end{equation}
Note that, since $b\in \ff{2^s}$ then also $c\in\ff{2^s}$ and $b'\in\ff{2^s}$. 

Now, suppose that $c=1$, then Equation \eqref{eqn:3} becomes
$$
x+x^{2^{2k}}={b'}^{2^{-1}}=b''.
$$
Thus, we can substitute $x^{2^{2k}}= x+b''$ in Equation \eqref{eqn:1} obtaining
\begin{equation}\label{eqn:4}
x^2+{b''}(x+x^{2^k})+x^{2^k}=b+b''.
\end{equation}

If we raise \eqref{eqn:4} by ${2^{k}}$ and we add it to \eqref{eqn:4}, we have
$$
(x+x^{2^k})^2+(b''+b''^{2^k}+1)(x+x^{2^k})=b''^{2^k+1}+b+b^{2^k}+b''^{2^k}.
$$
Noting that $b''+b''^{2^k}=  (x+x^{2^{2k}})+ (x+x^{2^{2k}})^{2^k}=c=1$ we have 
$$
x+x^{2^k}=({b''^{2^k+1}+b''^{2^k}+b+b^{2^k}})^{2^{-1}}=t.
$$
Now, substituting $x^{2^k}=x+t$ in \eqref{eqn:4} we obtain the equation
$$
x^2+x=b+b''+{b''}t+t.
$$
Since $b+b''+{b''}t+t\in \ff{2^s}$ we have that $x^2+x=x^{2^{s+1}}+x^{2^s}$. Therefore,  $x^{2^s}+x\in\ff{2}\subset\ff{2^s}$. This is not possible since it would imply $x=x^{2^{2s}}$ and $x\in \ff{2^n}\cap\ff{2^{2s}}=\ff{2^s}$.
Then, if $c=1$ we cannot have solutions $x\notin\ff{2^s}$.

Consider now, $c\ne 1$. Substituting $x=(c+1)y$ in \eqref{eqn:3} we obtain 
$$
(c+1)^2 [(y+y^{2^{2k}})^2+(y+y^{2^{2k}})]=b'.
$$
Note that $c\in\ff{2^k}\cap\ff{2^s}$ and then if $x\notin\ff{2^s}$ so do $y$. Now, we have that $y+y^{2^{2k}}$ is a solution of the equation $z^2+z=b'/(c+1)^2$. Then, $y+y^{2^{2k}}=t$ or $t+1$ for some $t \in \ff{2^s}\cap\ff{2^{2k}}$ (note that $b'\in\ff{2^s}$ and we cannot have solution in $\ff{2^n}\setminus\ff{2^s}$ for $z^2 +z=b'$).\\
If $y+y^{2^{2k}}=t$ then we can substitute $y^{2^{2k}}=y+t$ in Equation (\ref{eqn:1}) giving
\begin{equation} \label{eqn:8}
(c+1)^{2}((y+y^{2^k})t+y^2)+(c+1)(y^{2^k}+t)=b.
\end{equation}
Adding Equation \eqref{eqn:8} to itself raised by ${2^k}$ we get
\begin{eqnarray}
(c+1)^{2}(y+y^{2^k})^2 + [(c+1)^{2}(t+t^{2^k})+(c+1)](y+y^{2^k}) \nonumber \\
 +(c+1)^{2}t^{2^k+1}+(c+1)t^{2^{k}}+b+b^{2^k}&=&0.  \label{eqn:10}
\end{eqnarray}
Since $x=(c+1)y$ we have $c=\Tr^{4k}_k(x)=(c+1)\Tr^{4k}_k(y)$ and $\Tr^{4k}_k(y)=t^{2^k}+t$. Therefore, $t^{2^k}+t=c/(c+1)$ and 
\begin{eqnarray}
(c+1)^{2}((y+y^{2^k})^2 + (y+y^{2^k})) \nonumber \\
 +(c+1)^{2}t^{2^k+1}+(c+1)t^{2^{k}}+b+b^{2^k}&=&0.  \label{eqn:11}
\end{eqnarray}
Note that also considering $y+y^{2^{2k}}=t+1$ we would obtain the same equation. 

Now, from \eqref{eqn:11} we have that $y+y^{2^k}=r$ or $r+1$ for some $r\in\ff{2^s}$.

If $y^{2^k}=y+r$ then $y^{2^{2k}}=y+r+r^{2^k}$ and, substituting in (\ref{eqn:1}), we obtain 
$$
\begin{aligned}
(c+1)^{2}((y+r+r^{2^k})(y+r)+(y+r+r^{2^k})y+(y+r)y)+(c+1)(y+r^{2^k})=&\\
(c+1)^{2}y^2+(c+1)y+(c+1)^2(r^{2^k+1}+r^2)+(c+1)r^{2^k}=b\end{aligned}$$
which is the same as
$$x^2+x+(c+1)^2(r^{2^k+1}+r^2)+(c+1)r^{2^k}=b,$$
and then $x^2+x=d$ for some $d\in\ff{2^s}$, which is not possible as seen above.

If $y^{2^k}=y+r+1$, then we would obtain
$$x^2+x+(c+1)^2(r^{2^k+1}+r^{2^k}+r^2+r)+(c+1)(r+1)^{2^k}=b.$$
Then, Equation  \eqref{eqn:1} does not admit solutions in $\ff{2^n}\setminus\ff{2^s}$ when $b\in\ff{2^s}$. 
\qed\end{proof}

\begin{theorem}\label{th:6diff2}
Let $n=4k=sm$, with $k$, $m$ odd and $s$ even. Let $f$ be at most differentially 6-uniform permutation over $\ff{2^s}$. Then
$$F(x)=f(x)+(f(x)+x^{2^{2k}+2^k+1})(x^{2^s}+x)^{2^n-1}=\begin{cases}
f(x)&\mbox{ if $x\in\ff{2^s}$}\\
x^{2^{2k}+2^k+1}&\mbox{ if $x\notin\ff{2^s}$}
\end{cases}$$
is a differentially 6-uniform permutation over $\ff{2^n}$.
\end{theorem}
\begin{proof}
From Lemma~\ref{lm:4k} we have that $x^{2^{2k}+2^k+1}+(x+1)^{2^{2k}+2^k+1}=b$ has no solution in $\ff{2^n}\setminus\ff{2^s}$ when $b\in\ff{2^s}$. So the claim follows from Theorem~\ref{th:main}.
\qed\end{proof}

From Theorem~\ref{th:6diff1} and Theorem~\ref{th:6diff2} we obtain a general construction for piecewise functions with differential uniformity at most 6. In the following, we will show that using a function $f$ which is affine equivalent to the inverse function we can obtain a permutation of maximal degree $n-1$ with high nonlinearity.


We, first, give the following result.

\begin{lemma}\label{lm:deg}
Let $F$ be a function defined over $\ff{2^n}$. Then, $F$ in its polynomial representation has a term of algebraic degree $n-1$ if and only if there exists a linear monomial $x^{2^j}$ such that $\sum_{x\in\ff{2^n}}F(x)x^{2^j}\ne 0$.
\end{lemma}
\begin{proof}
We note that for any $i\le 2^n-1$ and $j\le n-1$ we have  
$$
w_2(i)+1\ge w_2(i+2^j)\ge 
\deg(x^{i+2^j} \mod x^{2^n}-x).
$$

Indeed, let us consider the first inequality. If $2^j$ is not in the binary expansion of $i$, that is $i=\sum_{k=0}^{n-1}b_k2^k$ with $b_j=0$, then $w_2(i+2^j)=w_2(i)+1$. If $b_j=1$, then, denoting by $h=\min\{n, k\,:\, j+1\le k\le n-1, \, b_k=0\}$, we have $i+2^j=\sum_{k=0}^{j-1}b_k2^k+2^h+ \sum_{k=h+1}^{n-1}b_k2^k$, implying $w_2(i+2^j)=w_2(i)-(h-j)+1< w_2(i)+1$.

For the second inequality, we have that if $i+2^j\le 2^n-1$, then $w_2(i+2^j)=\deg(x^{i+2^j} \mod x^{2^n}-x)=\deg(x^{i+2^j})$. If $i+2^j> 2^n-1$, then it means that $i=\sum_{k=0}^{j-1}b_k2^k+\sum_{k=j}^{n-1}2^k$ and $i+2^j=2^n+\sum_{k=0}^{j-1}b_k2^k$. So, denoting by $h=\min\{j, k\,:\, 0\le k\le j-1, \, b_k=0\}$, we have $$x^{i+2^j} \mod x^{2^n}-x=x\cdot x^{\sum_{k=0}^{j-1}b_k2^k}= x^{2^h+\sum_{k=h+1}^{j-1}b_k2^k},$$
and so $w_2(i+2^j)\ge \deg(x^{2^h+\sum_{k=h+1}^{j-1}b_k2^k})$.

Thus, since for any function $G$ we have $\sum_{x\in\ff{2^n}}G(x)\ne 0$ if and only if $\deg(G)=n$, we obtain that if $w_2(i)<n-1$, then $\sum_{x\in\ff{2^n}}x^{i+2^j}=0$.

Now, consider the function $F$. We can write $F(x)=cx^{2^n-1}+\sum_{k=0}^{n-1}a_kx^{2^n-1-2^k}+G(x)$, where $\deg(G)\le n-2$. Note that for any $k=0,...,n-1$ and $0\le j\le n-1$, we have 
$$
x^{2^n-1-2^k+2^j} \mod x^{2^n}-x\equiv \begin{cases}
x^{2^n-1}& \text{ if } j=k\\
x^{(2^{n}-1)-(2^{k-1}+...+2^{j})}& \text{ if } j <k\\
x^{2^{j-1}+...+2^k}& \text{ if } j> k.
\end{cases}
$$

Let $\bar k=\min\{k\,:\, a_k\ne 0\}$, we have that 
$$\begin{aligned}
\sum_{x\in\ff{2^n}}F(x)x^{2^{\bar k}}=&c\sum_{x\in\ff{2^n}}x^{2^{\bar k}}+a_{\bar k}\sum_{x\in\ff{2^n}}x^{2^n-1}+\sum_{k=0}^{\bar k-1}\sum_{x\in\ff{2^n}}a_kx^{2^{\bar k-1}+...+2^k}\\
&+\sum_{k=\bar k+1}^{n-1}\sum_{x\in\ff{2^n}}a_kx^{2^n-1-(2^{k-1}+...+2^{\bar k})}+\sum_{x\in\ff{2^n}}G(x)x^{2^{\bar k}}=a_{\bar k}.
\end{aligned}
$$
Thus, if $F(x)$ has a term of algebraic degree $n-1$, that is, there exists $a_k\ne 0$, then we can find a monomial $x^{2^j}$ for which $\sum_{x\in\ff{2^n}}F(x)x^{2^{j}}\ne 0$. 

Vice versa, from above we have that if, by contradiction, $F(x)=cx^{2^n-1}+G(x)$, with $\deg(G)\le n-2$, then $\sum_{x\in\ff{2^n}}F(x)x^{2^{j}}= c\sum_{x\in\ff{2^n}}x^{2^{j}}+\sum_{x\in\ff{2^n}}G(x)x^{2^{j}}=0$ for any $j$.
\qed\end{proof}
\begin{remark}
If $F$ is a permutation, Lemma~\ref{lm:deg} gives a necessary and sufficient condition for having $\deg(F)=n-1$. Moreover, it is possible to generalize Lemma~\ref{lm:deg} also for the case  $\deg(F)=n-t$, using monomials $x^d$ with $w_2(d)=t$. This result is similar to that given in \cite{PHT17,TCT}, where the authors use an $n$-variables Boolean function to give a sufficient condition for having $\deg(F)=n-1$. In particular, they show that if there exists a Boolean function $h(x)$ of degree $n-k$ such that $\sum_{x\in\ff{2^n}}F(x)h(x)\ne 0$, then $\deg(F)\ge k$.
\end{remark}

\begin{corollary}\label{cor:1}
Let $n=sm$ with $s$ even such that $s/2$ and $m$ are odd. Let $k$ be such that $\gcd(k,n)=2$ and $f(x)=A_1\circ Inv\circ A_2(x)$, where $Inv(x)=x^{-1}$ and $A_1,A_2$ are affine permutations over $\ff{2^s}$. Then
$$F(x)=f(x)+(f(x)+x^{2^k+1})(x^{2^s}+x)^{2^n-1}=\begin{cases}
f(x)&\mbox{ if $x\in\ff{2^s}$}\\
x^{2^k+1}&\mbox{ if $x\notin\ff{2^s}$}
\end{cases}$$
is a differentially 6-uniform permutation over $\ff{2^n}$. Moreover, if $s>2$ then the algebraic degree of $F$ is $n-1$.
\end{corollary}
\begin{proof}
We need to prove only that the degree of $F$ is $n-1$. From Lemma~\ref{lm:deg}, 
since $f(x)$ is affine-equivalent to the inverse function over $\ff{2^s}$, there exists a monomial $h(x)=x^{2^j}$ in $\ffx{2^s}$ (with $j\le s-1$) such that 
$$
\sum_{x\in\ff{2^s}}f(x)h(x)\ne 0.
$$

Thus, since $x^{2^k+1}$ has algebraic degree equal to $2<s-1$ over $\ff{2^n}$ and also over $\ff{2^s}$, we obtain
$$
\begin{aligned}
\sum_{x\in\ff{2^n}}F(x)h(x)&=\sum_{x\in\ff{2^s}}f(x)h(x)+\sum_{x\in\ff{2^n}}x^{2^k+1}h(x)+\sum_{x\in\ff{2^s}}x^{2^k+1}h(x)\\&=\sum_{x\in\ff{2^s}}f(x)h(x)\ne 0.
\end{aligned}$$
Then, we have a term of degree $n-1$ in the polynomial representation of $F(x)$, implying $\deg(F)= n-1$ since $F$ is a permutation.
\qed\end{proof}
\begin{remark}
When $s=2$, we have that $x^{2^k+1}$ has algebraic degree $s-1=1$ over $\ff{2^s}$, and we could obtain a permutation $F(x)=f(x)+(f(x)+x^{2^{2k}+2^k+1})(x^{2^s}+x)^{2^n-1}$ of degree less than $n-1$ (an example is provided for $n=6$ and $n=10$ in Table~\ref{tab:6} and Table~\ref{tab:10}).
\end{remark}

Similarly we have the following construction using the Bracken-Leander function.

\begin{corollary}\label{cor:2}
Let $n=4k=sm$ with $k$, $m$ odd and $s$ even. Let $f(x)=A_1\circ Inv\circ A_2(x)$, where $Inv(x)=x^{-1}$ and $A_1,A_2$ are affine permutations over $\ff{2^s}$. Then
$$F(x)=f(x)+(f(x)+x^{2^{2k}+2^k+1})(x^{2^s}+x)^{2^n-1}=\begin{cases}
f(x)&\mbox{ if $x\in\ff{2^s}$}\\
x^{2^{2k}+2^k+1}&\mbox{ if $x\notin\ff{2^s}$}
\end{cases}$$
is a differentially 6-uniform permutation over $\ff{2^n}$. Moreover, if $s>4$ then  $\deg(F)=n-1$.
\end{corollary}

\begin{remark}
For $s=4$, we have that $x^{2^{2k}+2^k+1}$ has algebraic degree $s-1=3$ over $\ff{2^s}$, and we could obtain a permutation $F(x)=f(x)+(f(x)+x^{2^{2k}+2^k+1})(x^{2^s}+x)^{2^n-1}$ of degree less than $n-1$ (we will provide an example for $n=12$ in Table~\ref{tab:12}).
\end{remark}

\begin{proposition}\label{prop:walsh}
The nonlinearity of the functions in Corollary~\ref{cor:1} and Corollary~\ref{cor:2} is at least $2^{n-1}-2^{\frac{n}{2}}-2^{\frac{s}{2}+1}$.
\end{proposition}
\begin{proof}
Consider any function 
$$
F(x)=\begin{cases}
f(x)&\mbox{ if $x\in\ff{2^s}$}\\
g(x)&\mbox{ if $x\notin\ff{2^s}$}
\end{cases}$$
where $f$ and $g$ have coefficients in $\ff{q^s}$.

Then the Walsh coefficient $\cW_F(a,b)$ satisfies the following
$$
\begin{aligned}
\cW_F(a,b)=&\sum_{x\in\ff{2^s}}(-1)^{Tr(ax+bf(x))}+\sum_{x\in\ff{2^n}\setminus\ff{2^s}}(-1)^{Tr(ax+bg(x))}\\
		=&\sum_{x\in\ff{2^s}}(-1)^{Tr(ax+bf(x))}+\sum_{x\in\ff{2^n}}(-1)^{Tr(ax+bg(x))}-\sum_{x\in\ff{2^s}}(-1)^{Tr(ax+bg(x))}\\
		=&\sum_{x\in\ff{2^s}}(-1)^{Tr^s_1(Tr^n_s(a)x+Tr^n_s(b)f(x))}+\sum_{x\in\ff{2^n}}(-1)^{Tr(ax+bg(x))}\\
		&-\sum_{x\in\ff{2^s}}(-1)^{Tr^s_1(Tr^n_s(a)x+Tr^n_s(b)g(x))}.
\end{aligned}
$$
Then, 
$$\begin{aligned}
|\cW_F(a,b)|\le&|\sum_{x\in\ff{2^s}}(-1)^{Tr^s_1(Tr^n_s(a)x+Tr^n_s(b)f(x))}|+|\sum_{x\in\ff{2^n}}(-1)^{Tr(ax+bg(x))}|\\
		&+|\sum_{x\in\ff{2^s}}(-1)^{Tr^s_1(Tr^n_s(a)x+Tr^n_s(b)g(x))}|\\
		=& |\cW^{(s)}_f(Tr^n_s(a),Tr^n_s(b))|+|\cW^{(s)}_g(Tr^n_s(a),Tr^n_s(b))|+|\cW_g(a,b)|.
\end{aligned}$$
where $\cW^{(s)}$ denotes the Walsh transform computed over $\ff{2^s}$.
For the inverse function, Gold function and Bracken-Leander function the maximal values of the module of a Walsh coefficient is $2^{\frac{t}{2}+1}$ \cite{BL10,goldwalsh}, where $t$ is the degree of the field extension where we are computing  the Walsh coefficient. 

Thus, for $F$ as in Corollary~\ref{cor:1} or Corollary~\ref{cor:2} we have
$$
|\cW_F(a,b)|\le 2^{\frac{s}{2}+2}+2^{\frac{n}{2}+1},
$$
implying
$$
\mathcal{N}\hspace{-0.75 mm}\ell(F)\ge 2^{n-1}-2^{\frac{s}{2}+1}-2^{\frac{n}{2}}.
$$
\qed\end{proof}

\begin{remark}
We can note that, whenever we modify the Gold or Bracken-Leander function with a bijective vectorial Boolean function $f(x)\in\ffx{2^s}$, which is at most 6-uniform of algebraic degree $s-1$, we can obtain a 6-uniform permutation of maximal algebraic degree. Moreover, the nonlinearity of the obtained function is greater or equal to $2^{n-1}-\frac{\max_{a\in\ffs{2^s},b\in\ff{2^s}}|\cW^{(s)}_f(a,b)|+2^{\frac{s}{2}+1}}{2}-2^{\frac{n}{2}}$.
\end{remark}

It is well known that the algebraic degree is not preserved by the CCZ-equivalence and in particular by the inverse transformation. However, for any permutation of maximal algebraic degree, from Lemma~\ref{lm:deg}, we have the following easy observation.

\begin{proposition}\label{prop:deginv}
Let $F$ be a permutation defined over $\ff{2^n}$. Then, $\deg(F)=n-1$ if and only if $\deg(F^{-1})=n-1$.
\end{proposition}
\begin{proof}
 Suppose $\deg(F)=n-1$. From Lemma~\ref{lm:deg} there exists a linear monomial $h(x)$ for which we have $\sum_{x\in\ff{2^n}}F(x)h(x)\ne 0$. Since $F$ is a permutation we have
$$
\sum_{x\in\ff{2^n}}F(x)h(x)=\sum_{x\in\ff{2^n}}xh(F^{-1}(x)),
$$
which implies $\deg(h\circ F^{-1})=n-1$. Since $h$ is a linear monomial we have that $\deg(F^{-1})=n-1$.
\qed\end{proof}

From this result we have that also the compositional inverses of the functions given in Corollary~\ref{cor:1} and Corollary~\ref{cor:2} are differentially 6-uniform functions with high nonlinearity and algebraic degree $n-1$. \\

In Table \ref{tab:6} and Table  \ref{tab:10}, we give the CCZ-inequivalent functions that can be obtained by Corollary~\ref{cor:1} for $n=6,10$. Let $\ffs{2^n}=\langle \gz\rangle$, we denote by $\go=\gz^{\frac{2^n-1}{3}}$ the primitive element of $\ff{4}$. Since over $\ff{4}$ all the permutations are linear (affine) we consider $f(x)=A\circ Inv$. In both cases, we have 5 CCZ-inequivalent functions. One is the Gold function $x^5$, obtained for $A(x)=x$ and it is differentially 4-uniform, three of these were obtained also in \cite{6diffe} and one is new. Note that the new function has degree $n-2$, indeed as pointed out in the proof of Corollary \ref{cor:1} when $s=2$ we could obtain functions with degree lower that $n-1$.

\begin{table}[h!]\caption{CCZ-inequivalent permutations from Corollary ~\ref{cor:1} over $\ff{2^6}$}\label{tab:6}
\centering\begin{tabular}{|c|c|c|c|c|}
\hline
$A(x)$ & deg & $\mathcal{N}\hspace{-0.75 mm}\ell(G)$ & Bound on $\mathcal{N}\hspace{-0.75 mm}\ell$&$\gd$\\
\hline
$x$ & 2 &24& 20& 4\\
\hline
$x+\go $ & 4 &20& 20&6\\
\hline
$\go x^2+\go $ & 5 &20& 20&6\\
\hline
$\go x$ & 5 &22& 20&6\\
\hline
$\go^2x^2+\go $ & 5 &22& 20&6\\
\hline
\end{tabular}
\end{table}

\begin{table}[h!]\caption{CCZ-inequivalent permutations from Corollary ~\ref{cor:1} over $\ff{2^{10}}$}\label{tab:10}
\centering\begin{tabular}{|c|c|c|c|c|}
\hline
$A(x)$ & deg & $\mathcal{N}\hspace{-0.75 mm}\ell(G)$ & Bound on $\mathcal{N}\hspace{-0.75 mm}\ell$&$\gd$\\
\hline
 $x$ & 2 &480& 476&4\\
\hline
$x+\go $ & 8 &476& 476&6\\
\hline
$\go x^2 + \go$ & 9 &476& 476&6\\
\hline
$\go x$ & 9 &478& 476&6\\
\hline
$\go^2 x^2+\go$ & 9 &478& 476&6\\
\hline
\end{tabular}
\end{table}

In Table \ref{tab:12} we report some permutations constructed from Corollary~\ref{cor:2} for $n=12$ (in this case $s=4$ and $m=3$). As before, we denote by $\go=\gz^{\frac{2^n-1}{3}}$ the primitive element of $\ff{4}\subset \ff{2^s}$ and we consider $f(x)=A\circ Inv$ with $A$ affine permutations defined over $\ff{4}[x]$. Note that for $A(x)=x^2$ we obtain the Bracken-Leander permutation.

\begin{table}[h!]\caption{CCZ-inequivalent permutations from Corollary ~\ref{cor:2} over $\ff{2^{12}}$}\label{tab:12}
\centering\begin{tabular}{|c|c|c|c|c|}
\hline
$A(x)$ & deg & $\mathcal{N}\hspace{-0.75 mm}\ell(G)$ & Bound on $\mathcal{N}\hspace{-0.75 mm}\ell$&$\gd$\\
\hline
$x^2$ & 3 &1984& 1976&4\\
\hline
$x^2+1 $ & 8 &1976& 1976&6\\
\hline
$\go^2 x^2 + \go$ &11 &1976& 1976&6\\
\hline
$ x+\go$ & 11 &1978& 1976&6\\
\hline
$\go x^2$ & 11 &1980& 1976&6\\
\hline
\end{tabular}
\end{table}

\section{Remarks on the boomerang uniformity for some piecewise permutations}

Recently, the boomerang uniformity of some permutations has been studied in \cite{BOCA,calvil,Li19,mesn,chun,nian,lipow}. However, the analysis has been focused, principally, on the case of quadratic or power functions. Only in \cite{Li19} and in \cite{calvil}, it has been determined the boomerang uniformity of some piecewise functions obtained from the inverse function modified on the subfield $\ff{4}$.

In this section, we give some results about the boomerang uniformity of piecewise permutations as in Theorem \ref{th:main}.
%

It was noted in \cite{Li19}, that the entry $T_F(a,b)$ of the BCT  of a function $F$ can be given by the number of solutions $(x,y)$ of the system 
$$
\begin{cases}
F^{-1}(x+a)+F^{-1}(y+a)=b\\
F^{-1}(x)+F^{-1}(y)=b.
\end{cases}
$$

Moreover, for the functions $F$ and $F^{-1}$ we have that $T_F(a,b)=T_{F^{-1}}(b,a)$. So, the boomerang uniformity of $F$ is given by the maximum number of solutions of the system
\begin{equation}
\begin{cases}\label{1}
F(x+a)+F(y+a)=b\\
F(x)+F(y)=b.
\end{cases}
\end{equation}

Denoting $y=x+\alpha$ and $D_aF(x)=F(x)+F(x+a)$ , we obtain
\begin{equation}\label{2}
\begin{cases}
D_\alpha F(x+a)=b\\
D_\alpha F(x)=b.
\end{cases}
\end{equation}
Therefore, let 
$$\begin{aligned}\gb_F(a,b)=&|\{(x,y)\in\fn^2\,:\, (x,y) \text{ is a solution of \eqref{1}}\}|\\
=&|\{(x,\ga)\in\fn^2\,:\, (x,\ga) \text{ is a solution of \eqref{2}}\}|,\end{aligned}$$ the boomerang uniformity of $F$ is given by 
$$\gb_F = \max_{a,b\in\ffs{2^n}}\gb_F(a,b).$$

For piecewise permutations as in Theorem~\ref{th:main} we obtain the following result.
\begin{proposition}\label{prop:boom}
Let $n=sm$. Let $$F(x)=f(x)+(f(x)+g(x))(x^{2^s}+x)^{2^n-1}=\begin{cases}
f(x)&\mbox{ if $x\in\ff{2^s}$}\\
g(x)&\mbox{ if $x\notin\ff{2^s}$}
\end{cases}$$
be a permutation over $\ff{2^n}$ where $f,g$ are as in Theorem \ref{th:main}, with $g$ satisfying (H3). Then, for any $a,b\in\ffs{2^n}$ 
$$
\gb_F(a,b)\le\begin{cases}
\gb_f(a,b)+\gb_g(a,b)& \text{if $a,b\in\ffs{2^s}$}\\
\gb_g(a,b)& \text{if $a\in\ffs{2^s}$ and $b\notin\ff{2^s}$}\\
\gb_g(a,b)+4\cdot N+2& \text{if $a,b\notin \ff{2^s}$},
\end{cases}
$$
where $N=|\{\ga\notin\ff{2^s}\mid \, f(x)+g(x+\ga)=b \text{ has solutions $x\in \ff{2^s}$}\}|$.
\end{proposition}
\begin{proof}
Note that System \eqref{1} could be divided into sixteen systems depending on $x$, $x+a$, $y$ and $y+a$ lying in $\ff{2^s}$ or not. However, for four of these systems we can immediately see that they are not possible, such as
$$
\begin{cases}
f(x)+f(y)=b\\
f(x+a)+g(y+a)=b,
\end{cases}
$$
since it would mean that $x,x+a$ and $y$ are in $\ff{2^s}$ and $y+a$ is not. But then, also $a\in\ff{2^s}$ and so $y+a\in\ff{2^s}$, which is a contradiction. 

Therefore, the possible systems that we could obtain are twelve:
$$
(S1)\begin{cases}
f(x)+f(y)=b\\
f(x+a)+f(y+a)=b\\
x,y,a\in\ff{2^s},\\
\end{cases}\quad
(S2)\begin{cases}
g(x)+g(y)=b\\
g(x+a)+g(y+a)=b\\
x,y,x+a,y+a\notin\ff{2^s},
\end{cases}
$$
$$
(S3)\begin{cases}
f(x)+g(y)=b\\
f(x+a)+g(y+a)=b\\
x,a\in\ff{2^s},y\notin\ff{2^s},
\end{cases}\quad
(S4)\begin{cases}
g(x)+f(y)=b\\
g(x+a)+f(y+a)=b\\
y,a\in\ff{2^s},x\notin\ff{2^s},
\end{cases}
$$
$$
(S5)\begin{cases}
f(x)+g(y)=b\\
g(x+a)+f(y+a)=b\\
x,y+a\in\ff{2^s},a\notin\ff{2^s},
\end{cases}\quad
(S6)\begin{cases}
g(x)+f(y)=b\\
f(x+a)+g(y+a)=b\\
x+a,y\in\ff{2^s},a\notin\ff{2^s},
\end{cases}
$$
$$
(S7)\begin{cases}
f(x)+g(y)=b\\
g(x+a)+g(y+a)=b\\
x\in\ff{2^s},y,y+a,a\notin\ff{2^s},
\end{cases}\quad
(S8)\begin{cases}
g(x)+f(y)=b\\
g(x+a)+g(y+a)=b\\
y\in\ff{2^s},x,x+a,a\notin\ff{2^s},
\end{cases}
$$
$$
(S9)\begin{cases}
g(x)+g(y)=b\\
f(x+a)+g(y+a)=b\\
x+a\in\ff{2^s},x,y,y+a\notin\ff{2^s},
\end{cases}\quad
(S10)\begin{cases}
g(x)+g(y)=b\\
g(x+a)+f(y+a)=b\\
y+a\in\ff{2^s},x,x+a,y\notin\ff{2^s},
\end{cases}
$$
$$
(S11)\begin{cases}
f(x)+f(y)=b\\
g(x+a)+g(y+a)=b\\
x,y\in\ff{2^s},a\notin\ff{2^s},
\end{cases}\quad
(S12)\begin{cases}
g(x)+g(y)=b\\
f(x+a)+f(y+a)=b\\
x+a,y+a\in\ff{2^s},a\notin\ff{2^s}.
\end{cases}
$$
First of all, note that:
\begin{itemize}
\item $(x,y)\in\ff{2^s}\times (\ff{2^n}\setminus\ff{2^s})$ is a solution of (S3) if and only if $(y,x)$ is a solution of (S4);
\item $(x,y)\in\ff{2^s}\times (\ff{2^n}\setminus\ff{2^s})$ is a solution of (S5) if and only if $(y,x)$ is a solution of (S6);
\item $(x,y)\in\ff{2^s}\times (\ff{2^n}\setminus\ff{2^s})$ is a solution of (S7) if and only if $(y,x)$ is a solution of (S8);
\item $(x,y)\in\ff{2^s}\times (\ff{2^n}\setminus\ff{2^s})$ is a solution of (S7) if and only if $(x+a,y+a)$ is a solution of (S9);
\item $(x,y)\in\ff{2^s}\times (\ff{2^n}\setminus\ff{2^s})$ is a solution of (S7) if and only if $(y+a,x+a)$ is a solution of (S10);
\item $(x,y)\in\ff{2^s}\times \ff{2^s}$ is a solution of (S11) if and only if $(x+a,y+a)$ is a solution of (S12).
\end{itemize}
So, we need to count the number of solutions of (S1), (S2), (S3), (S5), (S7) and (S11).

Let $a,b \in \ffs{2^s}$. Then,  the first equation of (S3) cannot be satisfied (recall that $f$ acts on $\ff{2^s}$ and $g$ over $\ff{2^n}\setminus\ff{2^s}$). The same for (S5) and (S7), while for (S11) the second equation is not possible.
So, only (S1) and (S2) could admit solutions, and thus we have at most $\gb_f(a,b)+\gb_g(a,b)$ solutions, where $\gb_f(a,b)$ is computed over $\ff{2^s}$. 

Let $b \in \ffs{2^s}$ and $a \notin \ff{2^s}$. As before,  (S3), (S5) and (S7) have no solutions. The same for (S1) since $x+a$ is not in $\ff{2^s}$. Let us check the solutions of (S11). 

Denoting by $y=x+\ga$, we would obtain 
$$
\begin{cases}
D_\ga f(x)=b\\
D_\ga g(x+a)=b,
\end{cases}
$$
where $\ga\in\ff{2^s}$. So $x+a\in\ff{2^n}\setminus\ff{2^s}$ would be a solution for $D_\ga g(x+a)=b$ with $\ga$ and $b$ in $\ff{2^s}$, which is not possible from hypothesis (H3). Then, only (S2) admits solutions, implying at most $\gb_g(a,b)$ solutions.

Let $b \notin \ff{2^s}$ and $a \in \ffs{2^s}$. Then, the first equation in (S1) is not admissible. Similar for (S11). Moreover, (S5) does not admit solutions since $x$ should be in $\ff{2^s}$ and $x+a$ not. The same for (S7).

Consider now (S3). 
Denoting by $y=x+\ga$, then we would obtain 
$$
\begin{cases}
f(x)+g(x+\ga)=b\\
f(x+a)+g(x+a+\ga)=b.
\end{cases}
$$
From the proof of Theorem \ref{th:main} we can have only one solution of $f(x)+g(x+\ga)=b$ with $x\in\ff{2^s}$ ($\ga\notin\ff{2^s}$). Then, $x=x+a$, implying $a=0$. So, also in this case, only (S2) can admit solutions.

Now, let $a,b \notin \ff{2^s}$. Then, with same arguments as above, (S1) and (S11) are not admissible. Also (S3) is not possible since $x$ and $x+a$ should be both in $\ff{2^s}$. Then, we need to analize the solutions of  (S5) and (S7). Consider (S5) and denote $y=x+\ga$. System (S5) is given by
$$
\begin{cases}
f(x)+g(x+\ga)=b\\
g(x+a)+f(x+a+\ga)=b,
\end{cases}
$$
where $\ga\in\ff{2^n}\setminus\ff{2^s}$.
From the proof of Theorem \ref{th:main} we have at most one solution $x$ in $\ff{2^s}$ for $f(x)+g(x+\ga)=b$. This implies $x=x+a+\ga$, and thus $\ga=a$. So, we have only one possible pair $(x,x+\ga)$ which can be a solution of (S5), and so also for (S6).

Consider, now, (S7). We get the system
$$
\begin{cases}
f(x)+g(x+\ga)=b\\
g(x+a)+g(x+a+\ga)=b.
\end{cases}
$$
Since from Theorem \ref{th:main} we have at most one solution $x$ in $\ff{2^s}$ for $f(x)+g(x+\ga)=b$, the number of solutions of (S7) are at most $N=|\{\ga\notin\ff{2^s}\mid \, f(x)+g(x+\ga)=b \text{ has solutions $x\in \ff{2^s}$}\}|$. This implies, also, at most $N$ solutions for (S8), (S9) and (S10). To conclude, (S2) admits at most $\gb_g(a,b)$ solutions.
\qed\end{proof}
\begin{remark}
Since $f$ and $g$ are permutations for any $x\in\ff{2^s}$ we have a unique $\ga\notin\ff{2^s}$ such that $f(x)+g(x+\ga)=b$. In particular, when $s=2$, we can have at most $4$ solutions for (S7), and when it admits $4$ solutions, (S5) has no solution. Thus, for $s=2$, the boomerang uniformity of the functions as in Theorem \ref{th:main} is upper bounded by $\gb_g+16$.
\end{remark}

In the following we consider the case when $g(x)$ is a 4-uniform Gold function, and $f(x)=g(x)+\g$ for some $\g\in \ffs{2^s}$, i.e. $F(x)=g(x)+\g+\g(x^{2^s}+x)^{2^n-1}$. Since $g$ is a power function, $F(x)=g(x)+\g+\g(x^{2^s}+x)^{2^n-1}$ is equivalent to $F'(x)=g(x)+1+(x^{2^s}+x)^{2^n-1}=\frac{1}{\g}F(g^{-1}(\g)\cdot x)$, so we can suppose that $\g=1$. Note that, such functions are differentially 6-uniform from Theorem~\ref{th:6diff1}.


\begin{proposition}\label{cor:boom}
Let $n=sm$ with $s$ even such that $s/2$ and $m$ are odd. Let $k$ be an integer such that $\gcd(k,n)=2$ and $g(x)=x^{2^k+1}$. Let 
$$F(x)=x^{2^k+1}+1+(x^{2^s}+x)^{2^n-1}=\begin{cases}
x^{2^k+1}+1&\mbox{ if $x\in\ff{2^s}$}\\
x^{2^k+1}&\mbox{ if $x\notin\ff{2^s}$}.
\end{cases}$$
Then, for any $a,b\in\ffs{2^n}$ we have
$$
\gb_F(a,b)\le\begin{cases}
4& \text{if $a\in\ffs{2^s}$ or $b\in\ffs{2^s}$}\\
22& \text{if $a,b\notin \ff{2^s}$}.
\end{cases}
$$
\end{proposition}
\begin{proof}
First of all, we can note that (S1) and (S2) in this case coincide. So for the case $a,b\in\ffs{2^s}$ we have $\gb_g(a,b)$ solutions, where $\gb_g(a,b)$ is computed on $\ff{2^n}$, and the boomerang uniformity of the Gold function $x^{2^k+1}$, over $\ff{2^n}$, is equal to $4$ (see for instance \cite{BOCA}). 

System (S7) in this case is given by
$$
\begin{cases}
D_\alpha f(x)=b+1\\
D_\alpha f(x+a)=b,
\end{cases}
$$
which is equivalent to
$$
\begin{cases}
D_\alpha f(x)=b+1\\
D_aD_\alpha f(x)=1.
\end{cases}
$$
Since $f$ is quadratic the second derivative $D_aD_\ga f(x)=a^{2^k}\ga+a\ga^{2^k}$ does not depend on $x$. So, for System (S7) the equation $D_aD_\ga f(x)=1$ is satisfied if and only if $\ga\in\gr+a\ff{4}$, where $\gr$ is such that $a^{2^k}\gr+a\gr^{2^k}=1$. This, implies that we can have at most four $\ga$'s admitting solutions in $\ff{2^s}$ for the equation $D_\ga f(x)=b$.
\qed\end{proof}
%


For the functions given in Table \ref{tab:6} and \ref{tab:10}, we give the boomerang uniformity in Table \ref{tab:b1} and \ref{tab:b2}, respectively. 
From the analysis given in Proposition~\ref{prop:boom}, we need to compute, principally, the number of solutions of Systems (S2), (S5) and (S7). 

Note that for $A(x)=x+\go $, since $x^5=x^{-1}$ over $\ff{4}$, we obtain the function studied in Proposition~\ref{cor:boom}.

\begin{table}[h!]\caption{Boomerang uniformity of the CCZ-inequivalent permutations from Corollary~\ref{cor:1} over $\ff{2^6}$}\label{tab:b1}
\centering\begin{tabular}{|cccccc|}
\hline
$A(x)$: & $x$ & $x+\go $ & $\go x^2+\go $& $\go x$ & $\go^2x^2+\go$\\
\hline
 $\gb_F$:& 4 &12& 12& 16& 12\\
\hline
\end{tabular}
\end{table}



\begin{table}[h!]\caption{Boomerang uniformity of the CCZ-inequivalent permutations from Corollary~\ref{cor:1} over $\ff{2^{10}}$}\label{tab:b2}
\centering\begin{tabular}{|cccccc|}
\hline
$A(x)$: & $x$ & $x+\go $ & $\go x^2+\go $& $\go x$ & $\go^2x^2+\go$\\
\hline
 $\gb_F$:& 4 & 8 &8 & 8& 8\\
\hline
\end{tabular}
\end{table}

\section*{Acknowledgements}
The research of this paper was supported by Trond Mohn Foundation.

\providecommand{\bysame}{\leavevmode\hbox to3em{\hrulefill}\thinspace}
\providecommand{\MR}{\relax\ifhmode\unskip\space\fi MR }
\providecommand{\MRhref}[2]{%
  \href{http://www.ams.org/mathscinet-getitem?mr=#1}{#2}
}
\providecommand{\href}[2]{#2}

\end{document}